\documentclass[11pt]{article}
\usepackage[boxruled]{algorithm2e}
\usepackage{amsmath, amsthm, amssymb,amsthm}
\usepackage{graphicx}
\usepackage{epsf}
\usepackage{epsfig}
\usepackage{color}
\usepackage{amsmath}
\usepackage{mathrsfs}
\usepackage{algorithm2e}
\usepackage{latexsym}
\usepackage{verbatim}
\usepackage{appendix}

\newsavebox{\fmbox}

\newcommand{\R}{\mbox{\rm R}}
\newcommand{\Rplus}{\R^+}

\newtheorem{theorem}{Theorem}[section]
\newtheorem{lemma}[theorem]{Lemma}

\def\A{${\cal A}$}

\newtheorem{claim}{Claim}

\newcommand\eat[1]{}
\begin{document}

\title{Optimal Approximation Algorithms for Multi-agent Combinatorial Problems with Discounted Price Functions}
\author{
{\Large\em  Gagan Goel, Pushkar Tripathi and Lei Wang\thanks{College
of Computing, Georgia Institute of Technology, Atlanta, GA
30332--0280}
 }
}

\date{}
\maketitle
\begin{abstract}
Submodular functions are an important class of functions in combinatorial optimization which satisfy the natural properties of decreasing marginal costs.  The study of these functions has led to strong structural properties with applications in many areas. Recently, there has been significant interest in extending the theory of algorithms for optimizing combinatorial problems (such as network design problem of spanning tree) over submodular functions. Unfortunately, the lower bounds under the general class of submodular functions are known to be very high for many of the classical problems. 

In this paper, we introduce and study an important subclass of submodular functions, which we call discounted price functions. These functions are succinctly representable and generalize linear cost functions. In this paper we study the following fundamental combinatorial optimization problems: Edge Cover, Spanning Tree, Perfect Matching and Shortest Path, and obtain tight upper and lower bounds for these problems.  

The main technical contribution of this paper is designing novel adaptive greedy algorithms for the above problems. These algorithms greedily build the solution whist rectifying mistakes made in the previous steps.
\end{abstract}
\thispagestyle{empty}
\setcounter{page}{0}
\newpage

\section{Introduction}
Combinatorial optimization problems such as matching, shortest path and spanning tree have been extensively studied in operations research and computer science. In these problems we are given a ground set $E$ and a collection $\Omega$ of subsets of the ground set that is usually implicitly defined by some combinatorial propoerty. We are also given the cost of each element in $E$ and the price of a subset in $\Omega$ is defined to be the sum of costs of its constituents. The objective is to find a subset in $\Omega$ with the minimum price. This framework with linear cost functions has served as the simplified model for many practical problems.


However, linear cost functions do not always model the complex dependencies of the prices in a real-world setting. Often, they only serve as approximations to the original functions. Hence solutions to these problems are rarely useful in the practical realm where these simplifying assumptions breakdown.

Another feature that is observed in a real-world scenarios is the presence of multiple agents, with different price functions, who wish to collaborate to build a single combinatorial structure. For linear cost functions, it is easy to see that having multiple agents doesn't change the complexity of the original problem. This is not the case for more general cost functions.

Submodular functions provide a natural way to model the non-linear dependencies between prices. These functions are the discrete analog of convexity and exhibit well known economic properties such as economy of scale and decreasing marginal cost. There has been some recent work \cite{Zoya,GKTW, satoru} to study the approximability of combinatorial problems under such functions and tight upper and lower bounds are known for many problems in this very general framework. 

However, the submodular setting suffers from a major drawback. Although in real world applications an agent's cost function is not necessarily as simple as a linear function, it is often not as complicated as a general submodular function either. i.e. submodular functions seem to tackle the problem in too much generality which often makes problems harder than required. As a result most of the bounds presented in previous work in this area are weak and have little practical applicability. A barrier in designing efficient algorithms for submodular cost functions arises from difficulty in representing them succinctly. Since a general submodular function is defined over an exponentially sized domain it is customary to assume that the function can be accessed by an oracle which can be querried polynomial number of times. This limits our ability to design efficient algorithms for these problems owing to the incomplete information about the functions, resulting in information theretic lower bounds.

In this paper we study an important subclass of succinctly representable submodular functions called \textit{discounted price functions} in the multi-agent setting. These functions are motivated by the observation that often in practical settings agents offer discounts based on the amount of goods being purchased from them i.e. purchasing more goods from a single agent results in higher discounts. Apart from such practical considerations discount functions have strong theoretical motivation too. Recent work \cite{goemansLearning} pertaining to learning submodular functions has lead to succicnt representations for submodular functions. These representations also behave like discounted price functions.

\subsection*{Discounted Price Model}
We define a function $d: \Rplus \rightarrow \Rplus$ to be a discounted price function if it satisfies the following properties: (1) $d(0)=0$; (2) $d$ is increasing; (3) $d(x)\le x$ for all $x$; (4) $d$ is concave. 

We study combinatorial problems in the following general setting. We are given a set of elements, $E$, a collection $\Omega$ of its subsets. We are also given a set ${\cal A}$ of $k$ agents where each agent, $a \in {\cal A}$
specifies a cost function $c_a:E\rightarrow \Rplus$ where $c_a(e)$ indicates her cost for the element $e$. Each agent also declares a discounted price function, $d_a$. If an agent $a$ is assigned a set of elements $T$, then her total price is specified by $d_a(\sum_{e\in T}c_a(e))$. This is called her \textit{discounted price}. The objective is to select a subset $S$ from $\Omega$ and a partition $S_1,S_2,...,S_k$ of $S$, such that $\sum_{a \in {\cal A}}d_a(\sum_{e\in S_a}c_a(e))$ is minimized.

Our frame work generalizes the classical single-agent linear cost setting. On the other hand, (4) implies that discounted price functions are satisfy the decreasing marginal cost property. Moreover, since any concave function can be well approximated by a piecewise linear function, discounted price functions can be viewed as a class of succinctly represented submodular functions.

We study the following four problems over an undirected complete graph $G = (V,E)$.

\begin{itemize}
\item \textbf{Discounted Edge Cover}: In this problem, $\Omega$ is chosen to be the collection of edge covers.

\item \textbf{Discounted Spanning Tree}: In this problem, $\Omega$ is the collection of spanning trees of the graph.

\item \textbf{Discounted Perfect Matching}: In this problem, we assume the graph has an even number of vertices. $\Omega$ is chosen to be the collection of perfect matchings of the graph.

\item \textbf{Discounted Shortest Path}: In this problem, we are given two fixed vertices $s$ and $t$. $\Omega$ consists of all paths connecting $s$ and $t$.
\end{itemize}

\subsection*{Our Results}
We summarize our results in the following theorem:
\begin{theorem}
There is an $O(\log n)$ approximate algorithm for each of the following problems: discounted edge cover,spanning tree,perfect matching and shortest path, where $n$ is the number of vertices in their instances. Moreover, it is hard to approximate any of them within factor $(1-o(1))\log n$ unless \emph{NP}=\emph{DTIME}$(n^{O(\log \log n)})$.
\end{theorem}

Our main technical contribution is the introduction of a new algorithmic technique called \textit{adaptive greedy algorithms}. In traditional greedy algorithms like set cover, we make the most 'efficient' decision at every step to inch closer to the objective. However, these decision are irreversible and cannot be modified in subsequent iterations. For example in the greedy set cover algorithm, if a set is chosen in a particular iteration it will surely be included in the solution. This approach works well for problems where the solution does not need to satisfy combinatorial properties, but fails in the combinatorial setting. This is because, choosing a particular set may preclude the choice of other efficient sets in future iterations owing to the combinatorial constraints. To rectify this problem we allow our algorithm to undo some of the decisions taken during previous iterations. We call such algorithms adaptive greedy algorithms. Another important contribution of the paper is to link the shortest path problem to the minimum weight perfect matching problem by a factor preserving reduction. 

From the hardness perspective, we consider the \emph{discounted reverse auction} problem in which the collection is devoid of any combinatorial constraints. Using a factor preserving reduction from the set cover problem we prove that it is hard to approximate the discounted reverse auction within factor $(1-o(1))\log n$ unless \emph{NP}=\emph{DTIME}$(n^{O(\log \log n)})$. We use discounted reverse auction as a starting point to prove the hardness of all the combinatorial problems considered in this paper. 

From the above argument, one would expect that our combinatorially structured problems might have weaker lower bounds than discounted reverse auction. However, interestingly, we find it is not the case, as shown by the $log$ approximate algorithms in the section \ref{algo section}.


\subsection*{Related work}
The classical versions of edge cover, spanning tree, perfect
matching and shortest path are well studied and polynomial time
algorithms are known for all these problems, see \cite{tree1, tree2,
matching1, matching2, path, edgecover}. These have served as part of the 
most fundamental combinatorial optimization model in the
development of computer science and operation research.

Recently, Svitkina and Fleischer \cite{Zoya} generalized the linear
cost settings of some combinatorial optimization problems such as
sparsest cut, load balancing and knapsack to a submodular cost
setting. They gave $\sqrt{n/\log n}$ upper and lower bounds for all
these problems. \cite{satoru} also studied combinatorial optimization 
problems in the single agent setting.

Multi-agent setting was first introduced in \cite{GKTW}. In their model of \emph{combinatorial problems with multi-agent
submodular cost functions} they gave $\Omega(n)$ lower bounds for
the problems of submodular spanning tree and perfect matching, and
gave a $\Omega(n^{2/3})$ lower bound for submodular shortest path.
They also gave the matching upper bounds for all these problems. 
We remark that the lower bounds presented in \cite{Zoya,GKTW}
are information theoretic and not computational.

Similar generalization from linear objective function
to the more general submodular function was applied to
maximization problems in \cite{maximizing, matroid1, matroid2}. The
multi-agent generalization of maximization problems which corresponds to
\emph{combinatorial auction} has been extensively studied both in computer science and economics \cite{CA,AGT} and tight information
lower bounds are known for these problems, see \cite{tight}.

\section{Algorithms for Discounted Combinatorial Optimization}
\label{algo section}
In this section we present approximation algorithms for the combinatorial problems defined earlier. In each of the problems we are given a complete graph $G=(V,E)$ over $n$ vertices. We are also given a set \A, of $k$ agents each of whom specifies a cost $c_a: E \rightarrow \Rplus$. Here $c_a(e)$ is the cost for building edge $e$ for agent $a$. Each agent also specifies a discounted price function given by $d_a: \Rplus \rightarrow \Rplus$. The objective is
to build a prespecified combinatorial structure using the edges in $E$, and allocate these edges among the agents such that the sum of discounted prices for the agents is minimized.


\subsection{Discounted Edge Cover} \label{edge cover}
In this section, we study the discounted edge cover problem and
establish a $\log n$ approximate algorithm for this problem.

Here is the basic idea behind our algorithm. Given a discounted edge
cover instance, we construct a set cover instance such that: (1) an
optimal edge cover corresponds to a set cover with the same cost and
(2) a set cover corresponds to an edge cover with a smaller price.
In the set cover instance, we apply the greedy algorithm from
\cite{setcoverAlgo} to get a set cover whose cost is within $\log n$ of
the optimal cost, then the corresponding edge cover gives an $\log
n$ approximation of the optimum edge cover. We remark that we will
have exponentially many sets in the set cover instance that we
construct for our problem. To apply the greedy algorithm, we need to
show that in each step, the set with the lowest \textit{average
cost} can be found in polynomial time.

Now we state our algorithm formally. Consider a set cover instance
where we have to cover the set of vertices, $V$, with $k2^n$ subsets
which are indexed by $(a,S) \in {\cal A} \times 2^V$. The cost of
the set $(a,S)$, denoted by $cost(a,S)$, is defined as the minimum
discounted price of an edge cover for the vertices in $S$ that can
be built by agent $a$. For the instance of set cover described
above, we apply the greedy algorithm \cite{setcoverAlgo}
 to get a set cover ${\cal S}$. Let $U_a$ be the set of vertices covered by sets of the form $(a,S) \in {\cal
S}$. Let $C_a$ be agent $a$'s optimal discounted edge cover of the
vertices in $U_a$. We output $\{C_a: a \in {\cal A} \}$ as our
solution.

The correctness of the algorithm follows from the observation that
each $C_a$ is a cover of $S_a$ thus their union must form an edge
cover of the vertex set $V$.

Now we show that the running time of the algorithm is polynomial in
$k$ and $n$. Notice that when $U_a$ is given, since agent $a$'s
discounted price function is increasing, $C_a$ can be found in
polynomial time. Hence we only need to show that the greedy
algorithm on our set cover instance can be done in polynomial time.

Recall that the greedy algorithm from \cite{setcoverAlgo} covers the ground set, $V$, in phases. Let $Q$ be
the set of covered vertices in the beginning of a phase and the
\textit{average cost} of a set $(a,S)$ is defined as $\alpha_a(S) =
cost(a,S)/|S-Q|$. The algorithm picks the set with the smallest
average cost, covers the corresponding vertices in that set and go
to the next phase until all the vertices are covered. To show that
this algorithm can be implemented efficiently, we only need to show the following:

\begin{lemma}\label{most efficient subset}
For any $Q\subset V$, we can find min$\{cost(a,S)/|S-Q|:(a,S)\in
{\cal A} \times 2^V\}$ in polynomial time.
\end{lemma}

\begin{proof}
We can iterate over all choices of agent $a \in {\cal A}$, thus the
problem boils down to finding min$\{cost(a,S)/|S-Q|: S\subseteq V\}$
for each $a\in {\cal A}$.

For each integer $d$, if we can find min$\{cost(a,S):|S-Q|=d\}$ in
polynomial time, then we are done since then we can just search over
all the possible sizes of $S-Q$. Unfortunately, it is
\emph{NP}-hard to compute min$\{cost(a,S):|S-Q|=d\}$ for all
integer $d$. We will use claim \ref{d edge cover} to circumvent this problem.

\begin{claim}\label{d edge cover}
For any graph $G = (V,E)$ and $Q \subseteq V$ and for any positive
integer $d$, we can find the set $(a,S)$ minimizing $cost(a,S)$ such
that $|S-Q|$ is at least $d$, in polynomial time.
\end{claim}

\begin{proof}
Refer to Appendix \ref{d edge cover proof}.
\end{proof}

By claim \ref{d edge cover} above we can generate a collection of
subsets $\{S_i \subseteq V : 1 \leq i \leq n\}$, such that $(a,S_i)$
has the lowest value of $cost(a,S)$ among all sets $S$ which satisfy
$|S - Q|\geq i$.

\begin{claim}
min$\{cost(a,S)/|S-Q|: S\subseteq V\}=min\{cost(a,S_i)/|S_i-Q|:1\le
i\le n\}$.
\end{claim}
\begin{proof}
Let $\bar S$ be the set that has the minimum average cost with respect
to agent $a$. Suppose $|\bar S-Q|=d$. By our choice of $S_d$, we have
$|S_d-Q|\ge d=|\bar S-Q|$ and $cost(a,S_d)\le cost(a,\bar S)$. Therefore
we have $cost(a,S_d)/|S_d-Q|\le cost(a,\bar S)/|\bar S-Q|$, hence they
must be equal.
\end{proof}

By iterating over all $a\in {\cal A}$, we can find
min$\{cost(a,S)/|S-Q|:(a,S)\in {\cal A} \times 2^V\}$ in polynomial
time.
\end{proof}

Next we show that the approximation factor of our algorithm is $\log
n$. Let $OPT_{EC}$ and $OPT_S$ denote the costs of the optimal
solutions for the discounted edge cover instance and the
corresponding set cover instance respectively. Let $\{C_a: a \in
{\cal A} \}$ be the edge cover reported by our algorithm. For all $a
\in {\cal A}$ let
 $O_a$ be the set of vertices covered by
agent $a$ in the optimal edge cover. The sets $\{(a,O^a): a \in
{\cal A}\}$ form a solution for the set cover instance. Therefore
$OPT_S\le OPT_{EC}$.

Since we use the greedy set cover algorithm to approximate $OPT_S$,
we have 
	\[\sum_{a \in {\cal A}} d_a(\sum_{e\in C_a}c_a(e))\le (\log
n) OPT_S \le (\log n) OPT_{EC} \]

Thus we have the following theorem.
\begin{theorem}
\label{edge cover thm} There is a polynomial time algorithm which
finds a $\log n$-approximate solution to the discounted edge cover
problem for any graph over $n$ vertices.
\end{theorem}


\subsection{Discounted Spanning Tree} \label{spanning tree}
In this section, we study the discounted spanning tree problem and
establish a $\log n$ approximation algorithm for this problem.

Let us first consider a simple $O(\log^2n)$-approximation algorithm.
Observe that a spanning tree is an edge cover with the connectivity
requirement. If we apply the greedy edge cover algorithm in section
\ref{edge cover}, there is no guarantee that we will end up with a
connected edge cover. We may get a collection of connected
components. We can subsequently contract these components and run
the greedy edge cover algorithm again on the contracted graph. We
keep doing this until there is only one connected component. By this
method, we will get a connected edge cover, i.e. a spanning tree.

We now analyze the above algorithm. Let $OPT_{ST}$ be the price of
the minimum discounted spanning tree. After each execution of the
greedy edge cover algorithm, there is no isolated vertex, hence the
contraction decreases the number of vertices by at least a a factor of half,
therefore we will have to run the greedy edge cover algorithm
at most $O(\log n)$ times. Let $OPT^r_{EC}$ be the price of the
minimum edge cover for the graph obtained after the $r^{th}$
contraction and let $C_r$ be the edge cover that we produce for this
iteration. Using theorem \ref{edge cover thm}, the price of $C_r$ is
at most $(\log n)OPT^r_{EC}$. It is easy to see that $OPT^r_{EC}$ is
at most $OPT_{ST}$ for every $r$. Hence the price of $C_r$ is
bounded by $(\log n)OPT_{ST}$. Since there are at most $O(\log n)$
iterations, the price of the spanning tree produced by the above
algorithm is bounded by $O(\log^2 n)OPT_{ST}$.

We observe that greedy edge cover and contraction are two main steps
in the above algorithm. Intuitively, they are used to satisfy the
covering and connectivity requirements respectively. The algorithm
proceeds by alternately invoking these subroutines. Based on this
observation, our idea to get a $\log n$ approximation algorithm is
to apply the following \emph{adaptive greedy} algorithm: rather than
apply contraction after each complete execution of the greedy edge
cover, we interleave contraction with the iterations of the greedy
edge cover algorithm. After each iteration, we modify the graph and
motivate our algorithm to get a connected edge cover at the end.

Now we describe our adaptive greedy algorithm. For every agent $a$
and subset of vertices $S$ we define $cost(a,S)$ as the cost of the
optimal edge cover for $S$. We define the \textit{average cost} of a
set $(a,S)$ as $\alpha_a^S = cost(a,S)/|S|$. The algorithm proceeds
in phases and in each phase we have two steps, \textbf{search} and
\textbf{contraction}. In each phase $r$, during the search step we
find the set $(a_r,S_r)$ with the lowest average cost and set the
\emph{potential} of each vertex $v\in S_r$ as
$p(v)=\alpha_{a_r}^{S_r}$. The search step is followed by a
contraction step, where we modify the graph by contracting every
connected component in the induced subgraph of agent $a_r$'s optimal
edge cover for the set $S_r$. After this we begin the next phase.
The algorithm terminates when we have contracted the original graph
to a single vertex. For every agent, we find the set of all edges
assigned to her across all the search steps declare this as her bundle of assigned edges. 
Finally remove unnecessary edges from the set of assigned 
edges to get a spanning tree.

It is easy to see that we get a connected edge cover at the end of
the algorithm, which proves the correctness of the algorithm. To analyze the running time, we observe that 
there can be at most $n$ phases and by Lemma \ref{most
efficient subset}, each phase can be implemented in polynomial time. 
Hence the algorithm runs in polynomial time.

Next, we prove that the approximation factor of the algorithm is
$O(\log n)$. Let $OPT_{ST}$ be the price of the optimal solution of
the discounted spanning tree instance and let $\{T_a: a \in {\cal
A}\}$ be our solution. Let $V'$ be the set of contracted vertices we
produced during the algorithm. Number the elements of $V$ and $V'$
in the order in which they were covered by the algorithm, resolving
ties arbitrarily. Suppose $V=\{v_1,...,v_n\}$ and
$V'=\{z_1,...,z_{n'}\}$. Obviously, $n'\le n$.

It is easy to verify that $\sum_a d_a(T_a)\le \sum_i p(v_i)+\sum_j
p(z_j)$. Therefore we only need to bound the potentials of the
vertices in $V\cup V'$.

\begin{claim}
\label{potential}
$p(v_i)\le \frac{OPT_{ST}}{n-i+1}$ for any $i \in \{1 \cdots n\}$ and
$p(z_j)\le \frac{OPT_{ST}}{n'-j}$ for any $j \in \{1 \cdots n' \}$.
\end{claim}

\begin{proof}
Refer to Appendix \ref{potential proof}.
\end{proof}

From the above claim, we have $$\sum_a d_a(T_a)\le \sum_{1\le i\le
n} \frac{OPT_{ST}}{n-i+1} +\sum_{1\le j\le n'}
\frac{OPT_{ST}}{n'-j}\le (\log n+\log n')OPT_{ST}\le O(\log
n)OPT_{ST}$$

Therefore we have the following:
\begin{theorem}
\label{spanning tree thm}
There is a polynomial time algorithm which finds an $O(\log
n)$-approximate solution to the discounted spanning tree
problem for any graph with $n$ vertices.
\end{theorem}

\subsection{Discounted Perfect Matching}\label{matching}
In this section we provide a $\log n$ approximate algorithm for
the discounted (minimum) perfect matching problem.

The main difficulty in tackling matching is that unlike spanning
tree or edge cover, the union of perfect matchings for two
overlapping subsets of vertices may not contain perfect matching
over their union. So a greedy algorithm which iteratively matches
unmatched vertices cannot be used to solve this problem.

Here we provide an adaptive greedy algorithm for solving
this problem. In our algorithm we successively match unmatched
 vertices. Once a vertex gets matched, it remains matched through the
 course of the algorithm. However unlike traditional greedy algorithms,
 we provide the algorithm the ability to change the underlying matching
 for previously matched vertices.  The algorithm is adaptive
 in the sense that it alters some of the decisions taken in
 previous phases to reduce the price incurred in the current phase.

Let us explain the algorithm in greater detail now. The algorithm
runs in phases and in each phase it augments the set of matched
vertices by adding unmatched vertices and rematching some of the
previously matched vertices. At the end of each phase let $Z$ denote
the set of matched vertices and let $M$ be the underlying matching
covering $Z$. For every $a \in$ \A \ we define $c_a^M: E \rightarrow
\Rplus$ which attains value $0$ for all edges in $M$ and takes the
value $c_a(e)$ for all other edges $e \notin M$. For any $F
\subseteq E$ let $V(F)$ be the set of vertices in the graph induced
by $F$. For a given matching $M$, we define the \textit{average cost} of a
set $(a,F) \in {\cal A} \times 2^E$ with respect to $M$ as
$\alpha_a^M(F) = \ {c_a^M(F)}/{|V(F)-Z|}$.

In the beginning of every phase we pick the element $(a,F)$ from
${\cal A} \times 2^E$ with the lowest average cost such that there
exists a perfect matching over $Z \cup V(F)$ using edges from $M
\cup F$. We defer details of this subroutine until the proof of
lemma \ref{matching lemma}. For each vertex $v \in V(F)-Z$ we define
its potential $p(v) = \alpha_a^M(F)$. Finally we augment the set $Z$
with vertices in $V(F)$ and update the underlying matching covering
it. The algorithm terminates when all vertices have been matched,
i.e. belong to $Z$.

We will now show that the algorithm can be implemented in polynomial
time. For this we will need to prove the following lemma.

\begin{lemma}
\label{matching lemma} Given a matching $M$ and $Z = V(M)$, we can
find, in polynomial time, the element $(a,F)$ from ${\cal A} \times
2^E$ with the lowest average cost with respect to $M$ such that
there exists a perfect matching over $Z \cup V(F)$ using edges from
$M \cup F$.
\end{lemma}

\begin{proof}
Refer to Appendix \ref{matching lemma proof}.
\end{proof}

To analyze the approximation factor for the algorithm, we number the
vertices in $V$ in the order in which they are added to $Z$ and
resolve ties arbitrarily. Let $V=\{v_1, v_2 \cdots v_n\}$. Let $OPT$
be the price of the optimal solution. We have the following lemma to
bound the potential of every vertex.

\begin{lemma}
\label{setcover lemma} For each $i \in \left\{ 1,2 \cdots n
\right\},\ p(v_i) \leq OPT/(n-i+1)$.
\end{lemma}
\begin{proof}
Similar to the proof of Claim \ref{potential}.
\end{proof}

Since the price of our solution is at most the sum of the potentials for the vertices 
the following theorem follows immediately from lemma \ref{setcover lemma} 
\begin{theorem}
There exists a $\log n$ approximate algorithm for the discounted perfect matching problem.
\end{theorem}

\subsection{Shortest s-t Path}
In this section we consider the discounted shortest path problem
between two given vertices $s$ and $t$. We provide a $\log n$
approximation algorithm for this problem.

Unlike any of the previous problems, it seems to difficult to design a greedy
algorithm for the discounted shortest path problem. However using ideas
developed in section \ref{matching} and by a factor preserving
reduction we obtain a $\log n$ approximate algorithm this problem. Our reduction
also implies that the shortest path problem is easier than the
minimum weight perfect matching problem, and this technique can be
applied to very general models whereby any algorithm for perfect
matching problem can be used to solve the shortest path problem
while preserving the approximation factor.

First of all, we describe our reduction:
\begin{lemma}\label{reduction}
Let $A$ be a $\beta$-approximate algorithm for the perfect
matching problem, then we can get a $\beta$-approximation for the
shortest path problem using $A$ as a subroutine
\end{lemma}
\begin{proof}
Suppose we are given a graph $G=(V,E)$. Construct an auxiliary graph
$\bar{G}$ in the following way: Replace every vertex $v \in V$ by
$v'$ and $v''$ and add an edge connecting them. The cost of this
edge is zero for every agent. For every edge $uv \in E$ where $u, v
\notin \left\{s, t\right\}$ we replace $uv$ with the gadget shown in
figure \ref{gad1}.

\begin{figure}[htbp]
   \centering
        \includegraphics[width=0.23\textwidth]{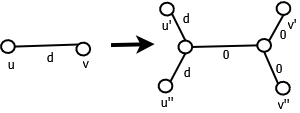}
        \caption{Gadget for edges not incident on $s$ or $t$}
    \label{gad1}
\end{figure}

If $d = c_a(uv)$ was the cost of the edge $uv$ for agent $a \in {\cal
A}$ then two of the edges in our gadget have the same cost as shown
above. Rest of the edges have zero cost for all agents. Replace any
edge $su \in E$ with the gadget shown in figure \ref{gad2}. Do the
same for any edge
$ut$.

\begin{figure}[htbp]
    \centering
        \includegraphics[width=0.30\textwidth]{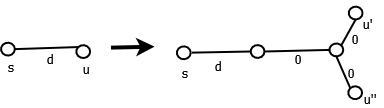}
        \caption{Gadget for edges incident on $s$ or $t$}
    \label{gad2}
\end{figure}

On this graph $\bar{G}$, use the algorithm $A$ to get the minimum
cost matching. Let $M$ be the matching returned. We can interpret
$M$ as a $s-t$ path in $G$ in the following way. Let $g(uv)$ be the edges in
$\bar{G}$ corresponding to the edge $uv$ for the gadget shown in
figure \ref{gad1}. Observe that either one or two edges of every
such gadget must belong to $M$. Let $S$ be the set of edges in $G$
such that two edges in their corresponding gadget belong to $M$. One
can check that every vertex in $V$ is incident with zero or two edges
from $S$, whereas $s$ and $t$ are each incident with exactly
one edge in $S$. Therefore $S$ consists of an $s-t$ path $P_S$ and
some other circuits. Now the circuits in $S$ must have cost zero.
This is because if a circuit has positive cost then the cost of the
matching can be reduced further by pairing up the vertices in the
circuit as shown in figure \ref{circuitFig}.

\begin{figure}[htbp]
    \centering
        \includegraphics[width=0.50\textwidth]{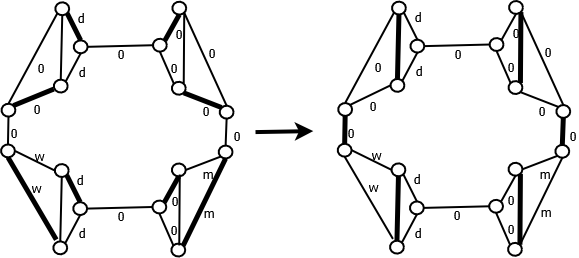}
        \caption{Circuits not involving edges in $S$ should have zero costs}
    \label{circuitFig}
\end{figure}

Note that this defines a cost preserving bijection between $s-t$
paths in $G$ to perfect matchings in $\bar{G}$.
\end{proof}

Using the algorithm from section \ref{matching}, by Lemma
\ref{reduction}, we have:

\begin{theorem}
There exists a $\log n$ approximate algorithm for the discounted
shortest path problem.
\end{theorem}

\section{Hardness of Approximation}
\label{lowerBounds}
In this section we present hardness of approximation results for the
problems defined earlier. Unlike some of the previous work on
combinatorial optimization \cite{GKTW,Zoya} over non-linear cost
functions, the bounds presented here are not information theoretic
but are contingent on $NP=DTIME(n^{O(\log\log n)})$. It should be noted that
the discount functions for the agents are the only part of the input
that is not presented explicitly. However it is well known
that concave functions that accessible through
an oracle can be closely approximated by a piecewise linear function
using \textit{few} queries to the oracle. So, we are not handicapped
by the lack of information from solving these problems optimally.

To show the hardness of approximation of the problems stated earlier
we consider the following general problem and use a reduction from
set cover to establish its hardness of approximation.

\textbf{Discounted Reverse Auction:} We are given a a set $E$ of $n$
items and a set \A \ of agents each of whom specifies a function $c_a:
E \rightarrow \Rplus$. Here $c_a(e)$ is the cost for procuring item
$e$ from agent $a$. Each agent also specifies a discounted price
function given by $d_a: \Rplus \rightarrow \Rplus$. The task is to 
find a partition ${\cal P} = \{P_1 \cdots P_k\}$ of $E$ such that 
$\sum_{a \in {\cal A}}d_a(\sum_{e\in P_a}c_a(e))$ is minimized.

\begin{lemma}
\label{hardness lemma}
It is hard to approximate the discounted reverse auction problem
within factor $(1-o(1))\log n$ unless $NP=DTIME(n^{O(\log\log n)})$.
\end{lemma}
\begin{proof}
Refer to Appendix \ref{hardness lemma proof}
\end{proof}

This reduction can be extended to other combinatorial problems in
this setting to give logarithmic hardness of approximation for many
combinatorial problems. This can be achieved by considering an
instance of the problem where we have just one combinatorial object
and our task is the allocate it optimally among the agents. For
example for the discounted spanning tree problem we consider the
instance when the input graph is itself a tree and we have to
optimally allocate its edges among the agents to minimize the total
price. Thus, we have the following:

\begin{theorem}
\label{edge cover lower bound} It is hard to approximate any of the
problems of the discounted edge cover,spanning tree,perfect matching
and shortest path can be approximated within factor
$(1-o(1))\log n$ on a graph over $n$ vertices unless $NP=DTIME(n^{O(\log\log n)})$.
\end{theorem}
\bibliography{discounted}
\bibliographystyle{plain}

\appendix
\subsection{Proof of Claim \ref{d edge cover}}
\label{d edge cover proof}
To find the desired set we construct a graph $G'=(V',E')$ as
follows: Add a set $X \cup Y$ to the set of vertices in $G$, where
$|X| = |Q|$ and $|Y| = n-d$. Match every vertex $X$ to a vertex $Q$
with an edge of cost $0$. Connect each vertex in $Y$ to each vertex in
$V$ by an edge of very large cost. Set the cost of each edge $e\in
E$ as $c_a(e)$. Find the minimum cost edge cover in $G'$. Let $S^*$
be such a cover. It is easy to verify that $S^*$ is the desired set.

\subsection{Proof of Lemma \ref{matching lemma}}
\label{matching lemma proof}
We may assume that the agent $a$, and the size of $|V(F)-Z|$ in the
formula for $\alpha_a^M(F)$ are fixed by iterating over all values
of these parameters. So the above statement may be viewed as one
concerning a single agent for a fixed value of $|V(F)-Z|$. Let this
agent be $a \in {\cal A}$ and let $t= |V(F)-Z|$. i.e. We wish to
find the minimum discounted weight matching which saturates $t$ new
vertices(not in $Z$) and keeps the vertices in $Z$ saturated. We
create an auxiliary graph $G'$ by adding $n - (|Z| + t)$ vertices to
$V$. We add edges connecting each of the newly added vertex with
each vertex in $V-Z$. Let $E'$ be the set of newly added edges. We
extend $c_a^M$ to $\hat{c}_a^M$ by setting the cost of edges in $E'$
to be $0$. Finally we find the minimum weight perfect matching in
$G'$ under the cost function $\hat{c}_a^M$. Let $M'$ be this
matching. It is easy to check that $M'-E'$ is the desired matching
saturating $t$ new vertices.

\subsection{Proof of Claim \ref{potential}}
\label{potential proof}
For $i \in \{1 \cdots n \}$, suppose $v_i$ is covered in phase $r$. Let $G^r$
be the underlying graph at the beginning of phase $r$. Since
$v_i,v_{i+1},...,v_n$ are not covered before phase $r$, $G^r$
contains at least $n-i+1$ vertices. Since the optimal spanning tree
can cover the vertices in $G^r$ by a price of $OPT_{ST}$, by our
greedy choice, $p(v_i)\le OPT_{ST}/(n-i+1)$.

Similarly, let $1\le j\le n'$ and assume $z_j$ is covered in phase
$r$. Since we should be able to produce $z_{j+1},z_{j+2},...,z_{n'}$
from contraction on vertices of $G^r$, there are at least $n'-r$
vertices in $G^r$. Therefore we have $p(z_j)\le OPT_{ST}/(n'-j)$.

\subsection{Proof of Lemma \ref{hardness lemma}}
\label{hardness lemma proof}
We reduce set cover to the discounted reverse auction problem to prove this result.

Consider an instance ${\cal I} = (U, C,w)$ of set cover where we
wish to cover all elements in the universe $U$ using sets from $C$
to minimize the sum of weights under the weight function $w: U
\rightarrow \Rplus$. We define an instance, ${\cal I'}$ of our
discounted reverse auction problem corresponding to ${\cal I}$ in
the following way. Let $U$ be the set of items. For every set $S \in
C$ define an agent $a_S$, whose cost function $c_a$ assigns the
value $w(S)$ for every element $s \in S$ and sets the cost of all
other elements in $U$ to be infinity. The discount price function
for the agent is shown in figure \ref{redFig}. Here the slope of the
second segment is small enough.

\begin{figure}[htbp]
    \centering
        \includegraphics[width=0.30\textwidth]{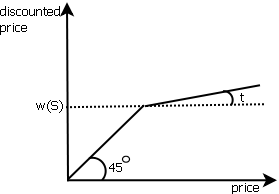}
        \caption{Discount function for agent corresponding to set $S$}
    \label{redFig}
\end{figure}

Consider a solution for ${\cal I'}$ where we procure at least one
item from agent $a_S$; then we can buy all elements in $S$ from
$a_S$ without a significant increase in our payment. So the cost of
the optimal solution to ${\cal I}$ can be as close to the price of
the optimal solution for ${\cal I'}$ as we want. By \cite{setcover}, 
set cover is hard to approximate beyond a factor of
$\log n$ unless $NP =DTIME(n^{O(\log\log n)})$. Therefore the discounted
reverse auction problem can not be approximated within factor
$(1-o(1))\log n$ unless $NP=DTIME(n^{O(\log\log n)})$.

\end{document}